\documentclass[11pt]{article}
\usepackage[T1]{fontenc}
\usepackage[latin9]{inputenc}
\usepackage{geometry}
\usepackage{color}
\usepackage{verbatim}
\usepackage{float}
\usepackage{mathtools}
\usepackage{algorithm2e}
\usepackage{amsmath}
\usepackage{amsthm}
\usepackage{amssymb}
\usepackage{setspace}
\usepackage{wasysym}
\usepackage{bbm}
\usepackage{esint}

\usepackage[inline]{enumitem}

\usepackage[authoryear]{natbib}
\usepackage[unicode=true,
 bookmarks=false,
 breaklinks=true,pdfborder={0 0 0},pdfborderstyle={},backref=false,colorlinks=true,citecolor=blue,urlcolor=blue]
 {hyperref}
\hypersetup{
 citecolor=blue, linkcolor=blue}

\makeatletter

\theoremstyle{plain}
\newtheorem{remark}{\protect\remarkname}
\theoremstyle{plain}

\theoremstyle{plain}
\newtheorem{proposition}{\protect\propositionname}
\theoremstyle{plain}
\newtheorem{theorem}{\protect\theoremname}
\theoremstyle{plain}

\theoremstyle{plain}
\newtheorem{definition}{\protect\definitionname}
\theoremstyle{plain}
\newtheorem{lemma}{\protect\lemmaname}
\advance \topmargin by -\headheight
\advance \topmargin by -\headsep

\usepackage[OT1]{fontenc}
\usepackage{amsthm,amsmath,amsfonts}
\usepackage{natbib}
\usepackage{booktabs}
\usepackage{mathrsfs}
\usepackage{graphicx,xcolor}
\usepackage{subcaption}
\usepackage{amssymb,latexsym}
\usepackage{textcomp}

\usepackage[scaled=0.92]{helvet}
\usepackage{times}
\usepackage[final]{microtype}
\usepackage{titling}

\LinesNumbered
\RestyleAlgo{algoruled}

\makeatother

\providecommand{\assumptionname}{Assumption}

\providecommand{\lemmaname}{Lemma}
\providecommand{\propositionname}{Proposition}
\providecommand{\remarkname}{Remark}
\providecommand{\theoremname}{Theorem}
\providecommand{\definitionname}{Definition}
\providecommand{\examplename}{Example}

\newtheorem*{remark*}{Remark}

\begin{document}

\title{CAIRO: Decoupling Order from Scale in Regression}

\author{Harri Vanhems \thanks{McGill University and Mila; Email:  \texttt{harri.vanhems@mail.mcgill.ca}.}
\and 
Yue Zhao \thanks{University of York; Email: \texttt{yue.zhao@york.ac.uk}.}
\and
Peng Shi \thanks{University of Wisconsin; Email: \texttt{pshi@bus.wisc.edu}.}
\and
Archer Y. Yang \thanks{Co-corresponding author. McGill University and Mila; Email: \texttt{archer.yang.yi@gmail.com}.}
}
\maketitle 

\begin{abstract}
Standard regression methods typically optimize a single pointwise objective, such as mean squared error, which conflates the learning of ordering with the learning of scale. This coupling renders models vulnerable to outliers and heavy-tailed noise. We propose \textsc{CAIRO} (\textbf{C}alibrate \textbf{A}fter \textbf{I}nitial \textbf{R}ank \textbf{O}rdering), a framework that decouples regression into two distinct stages. In the first stage, we learn a scoring function by minimizing a scale-invariant ranking loss; in the second, we recover the target scale via isotonic regression. We theoretically characterize a class of "Optimal-in-Rank-Order" objectives --- including variants of RankNet and Gini covariance --- and prove that they recover the ordering of the true conditional mean under mild assumptions. We further show that subsequent monotone calibration recovers the true regression function at the population level and mathematically guarantees that finite-sample predictions are strictly auto-calibrated. Empirically, CAIRO combines the representation learning of neural networks with the robustness of rank-based statistics. It matches the performance of state-of-the-art tree ensembles on tabular benchmarks and significantly outperforms standard regression objectives in regimes with heavy-tailed or heteroskedastic noise.

\end{abstract}

\section{Introduction}
\label{sec: intro}
Supervised learning models are typically evaluated against metrics that reflect specific deployment goals. We focus on two distinct families of such metrics. The first consists of regression-based metrics, like mean squared error (MSE), which reward predictions that are numerically close to the target value. The second consists of rank-based metrics, which reward models for correctly ordering pairs of examples relative to one another. While these objectives are often treated separately, a growing number of applications require models to succeed at both simultaneously.

In classification settings, this dual requirement is common, for instance, in ad click prediction or medical screening, where models must rank candidates effectively while also outputting calibrated probabilities. However, this need also rigorously extends to regression tasks with continuous, real-valued targets. A prime example is insurance pricing: a model must rank policyholders by risk to segment them correctly, but it must also output strictly \emph{auto-calibrated} premiums---meaning the expected actual loss for any predicted premium tier must exactly match the predicted premium ($\mathbb{E}[Y|\hat{m}(X)] = \hat{m}(X)$). Without this strict local unbiasedness, an insurer risks severe financial deficits.

While methodologies for combining calibration and ranking are well-established in classification, they are under-explored for regression with real-valued targets. Consequently, techniques designed for binary or discrete outcomes do not naturally generalize to continuous settings. In response, we introduce \textsc{CAIRO} (\textbf{C}alibrate \textbf{A}fter \textbf{I}nitial \textbf{R}ank \textbf{O}rdering), a two-stage framework designed to optimize both ranking and regression performance for real-valued targets. \textsc{CAIRO} learns sequentially: Stage 1 minimizes a ranking loss to learn the optimal ordering, and Stage~2 minimizes a pointwise loss to recover the scale. This framework is theoretically grounded; we show that it targets the true conditional expectation at the population level, just as a standard MSE-based learner would.

Crucially, the separation of learning into distinct stages provides two major statistical advantages over standard unconstrained regression. First, because ranking losses are scale-invariant, the representation learned in Stage~1 is inherently shielded from the magnitude of outliers. This allows \textsc{CAIRO} to function as a highly robust regression method, offering stability in regimes with heavy-tailed or heteroskedastic noise where standard MSE minimization often fails. Second, the isotonic projection in Stage~2 mathematically guarantees that the final predictions are \emph{auto-calibrated} on the training data. Standard neural networks optimized for global MSE frequently suffer from local miscalibration, systematically over- or under-predicting across different segments of the output space. By structurally enforcing that the empirical mean of the targets exactly matches the predicted value within every constant piece (level set) of the fitted function, \textsc{CAIRO} ensures locally unbiased estimation.

The remainder of this paper is organized as follows. Section~\ref{sec:background} reviews related work on learning-to-rank and score recalibration. Section~\ref{sec:our_method} introduces the \textsc{CAIRO} framework and formalizes its two-stage population objective. Section~\ref{sec: ranking_phase} studies Stage~1, including weighted pairwise ranking losses, their pointwise equivalents, and the class of \emph{Optimal-in-Rank-Order} (ORO) objectives that guarantee order-consistency with the conditional mean. Section~\ref{sec:calibration} analyzes Stage~2 and proves that monotone calibration via isotonic regression recovers the true regression function when the Stage~1 objective is ORO. Section~\ref{sec:calib-empirical} details empirical estimation and practical algorithms, including smooth surrogates for efficient optimization. Finally, Section~\ref{sec:experiments} presents synthetic and real-data benchmarks demonstrating that \textsc{CAIRO} is competitive with strong tree-based baselines and is particularly robust under heavy-tailed and heteroskedastic noise.

\section{Background and Related Work}
\label{sec:background}

\paragraph{Learning to rank.}

Traditionally, the ranking problem considers two i.i.d.\ copies $(X,Y)$ and $(X',Y')$, and constructs a \textit{score} $g$ leading to a rank rule: $g(X)>g(X') \implies Y>Y'$. The score $g$ can be obtained by minimizing the following pairwise ranking loss \citep{clemenccon2008ranking}:
\begin{equation}
\label{eq:01-pop}
\mathbb{P}\!\left\{(Y-Y')\,\operatorname{sign}\!\big(g(X)-g(X')\big)\le 0\right\},
\end{equation}
Such objectives are inherently scale-free which is different from regression tasks where the goal is to estimate the conditional mean
$m^\star(x) \equiv \mathbb{E}[Y| X=x]$ by minimizing a pointwise loss such as $\mathbb{E}\big[(Y-m(X))^2\big]$.

The learning-to-rank literature, especially in information retrieval, has developed a large array of ranking algorithms to estimate the optimal score $g$ given data $\{(x_i, y_i)\}_{i=1}^n$. Pairwise methods such as RankSVM \citep{joachims2002optimizing} and RankNet \citep{burges2005learning} optimize smooth surrogates of pairwise misordering, typically logistic or hinge-style losses applied to score differences $g(x_i)-g(x_j)$, often augmented with heuristic importance weights that reflect ranking metrics like NDCG or MAP \citep{burges2010ranknet}.  
Listwise methods such as ListNet, instead construct losses at the level of entire ranked lists, sometimes via differentiable relaxations of the sorting operator \citep{cao2007learning,taylor2008SoftRank}.  
More recent work introduces faster permutation-based relaxations such as \emph{SoftSort}, which approximate ranks and sorted scores with smooth operators suitable for gradient-based training \citep{prillo2020softsort}. 

\paragraph{Recalibration and auto-calibration.}
Many learning algorithms produce scores whose magnitude is not directly interpretable as it is generally not on the same numerical scale as the quantity one ultimately wishes to predict. Calibration refers to the task of learning a transformation of these raw outputs so that they match a desired target scale. In classification, the primary objective is to transform unnormalized outputs into reliable probability estimates via parametric mappings like Platt scaling \citep{platt1999probabilistic, niculescu2005predicting} or non-parametric approaches like histogram binning \citep{zadrozny2001obtaining, naeini2015obtaining}. In continuous regression, the analogous concept is \emph{auto-calibration} \citep{kruger2021generic, gneiting2023regression, wuthrich2023model}, which requires the prediction to be unbiased conditional on itself, i.e., $\mathbb{E}[Y|\hat{m}(X)] = \hat{m}(X)$. Isotonic regression offers a flexible non-parametric solution to achieve this by finding the optimal non-decreasing function that minimizes the squared error between scores and targets \citep{zadrozny2002transforming}. Because the Pool Adjacent Violators (PAV) algorithm computes block-wise local averages of the response, isotonic regression structurally guarantees empirical auto-calibration without assuming a rigid parametric form \citep{barlow1972statistical,guntuboyina2018nonparametric}.
 
\paragraph{Ranking for classification/regression.}
 
Combining ranking and calibration methods is an idea that has primarily been applied to classification to improve the accuracy of the class probability estimate. \citet{sculley2010combined, yan2022scale} proposed a joint objective approach that employs a single, multi-objective loss function to simultaneously target ranking and calibration. However, this approach forces a trade-off between the ranking and calibration objectives without theoretical guarantee that the calibrated outcome recovers the true class probability. This problem was overcome by \citet{menon2012predicting}, who proposed a sequential approach, learning the ranking first and calibration second. While their methods has theoretical guarantee to recover the true class probability, their analysis was limited to linear scores and binary outcomes, restricting its reach in machine learning applications.

\citet{wuthrich2024isotonic} recently adapted the sequential approach of \citet{menon2012predicting} to regression with real-valued outcomes. While they analyze the calibration of scores that ostensibly preserve the ordering of the conditional mean, their work does not characterize the class of loss functions guaranteed to produce such scores. Consequently, they do not provide theoretical assurances that the final calibrated estimator converges to the true regression function.  Furthermore, their implementation employs the Gamma deviance loss to train first-stage scores, yet the theoretical link between this loss and optimal ranking is unaddressed. They also restrict their analysis to linear scoring models.


The philosophy of prioritizing order over scale has recently demonstrated significant success in mechanism design, specifically through the \emph{Isotonic Mechanism} proposed for peer review \citep{su2021you, yan2025isotonic, wu2025isotonic}. These works establish that utilizing a provided ground-truth ranking to calibrate noisy scores can dramatically reduce estimation error.  Our work extends this insight into the supervised learning setting, though with a fundamental distinction in the source of the ranking. While the Isotonic Mechanism relies on eliciting a ranking from a human agent (e.g., an author ranking their own papers), CAIRO must learn a ranking function directly from feature representations. We demonstrate that the robustness benefits observed in mechanism design hold true in deep regression, even when the ranking itself is a learned latent construct rather than an external oracle.

\paragraph{Semiparametric single-index models.}

From a statistical perspective, modeling the conditional mean as $\mathbb{E}[Y| X=x] = f(g(x))$, where $f$ is an unknown monotonic link function and $g$ is a scoring index, naturally connects our framework to semiparametric single-index models (SIMs) \citep{hardle1993optimal,hardle1989investigating,ichimura1993semiparametric}. Our approach shares  conceptual roots with the maximum rank correlation (MRC) estimator \citep{han1987non}, which recovers a linear index $g(x) = x^\top \beta$ by maximizing Kendall's $\tau$. \textsc{CAIRO} thus serves as a flexible, deep-learning-compatible generalization of classical index models. However, it significantly departs from standard SIM estimation strategies. While traditional SIMs rely on computationally heavy profile-likelihoods or bandwidth-sensitive kernel smoothing to jointly estimate $f$ and $g$, \textsc{CAIRO} decouples the problem. By learning a highly non-linear index via distribution-free ranking losses and subsequently recovering the shape-constrained link through efficient isotonic regression, we circumvent classical computational bottlenecks.

\paragraph{Our contributions.} 

In this paper, \begin{enumerate*}[label=(\roman*)] 
\item we propose \textsc{CAIRO}, a two-step regression framework that extends the rank-then-calibrate approach of \citet{menon2012predicting} to real-valued outcomes. To ensure that the two-stage approach recovers the true regression function at the population level, we formalize the requirement that we expect of the first stage ranking loss and group them under a single collection of \emph{Optimal-in-Rank-Order} (ORO) loss functions. 
\item We provide   theoretical insights into these ranking losses by explicitly connecting them to popular measures of association such as Spearman's $\rho$ and Kendall's $\tau$. 
\item We mathematically ensure that our final estimator is \emph{auto-calibrated} via Stage~2 isotonic projection, effectively turning uncalibrated neural network representations into reliable, locally unbiased statistical forecasters.
\item Finally, we demonstrate empirically that the decoupling of scale and order allows \textsc{CAIRO} to function as a highly robust alternative to MSE minimization, exhibiting significant resilience to heavy-tailed and heteroskedastic noise. \end{enumerate*}

\section{The Decoupled Regression Framework}
\label{sec:our_method}
Consider a two-stage population learning framework to construct a regression model.
In Stage~1, a model is trained using a ranking loss $\mathcal{L}_\textup{rank}$ over a function class $\mathcal{G}$ (that is assumed to enclose the class $\mathcal{G}'$ in Definition~\ref{def:ORO}): 
\begin{align}\label{eq:stage1}
    \text{Stage 1:}\qquad g^{\star} &\in \arg\min_{g \in \mathcal{G}} \mathcal{L}_\text{rank}(Y, g(X)).
\end{align}
The purpose of Stage~1 is to learn a strictly monotone transformation of the true regression function $m^{\star}$. Specifically, later, we will show that we expect the optimal score from Stage~1 to be of the form $g^{\star}(X) = h(m^{\star}(X))$ for some $h$ strictly increasing. In Stage~2, the output from Stage~1 is calibrated with isotonic regression to recover $m^{\star}$.  Specifically, we apply isotonic regression in Stage~2 which minimizes the squared loss $\mathcal{L}_\textup{iso}$ over the space of non-decreasing functions $\mathcal{F}$ while preserving the ranking learned in Stage~1. Formally,
\begin{align}\label{eq:stage2}
    \text{Stage 2:}\qquad f^{\star} &\in \arg\min_{f \in \mathcal{F}} \mathcal{L}_\text{iso}(Y, f(g^{\star}(X))).
\end{align}
The isotonic regression step in Stage~2 yields a model $f^{\star}$ that can be viewed as learning the inverse of the $h(\cdot)$ transformation in $g^{\star}(X) = h(m^{\star}(X))$, thereby recovering the true regression function. We will show that $f^{\star}(g^{\star}) = m^{\star}$.  Henceforth, we refer to this two-stage framework as the Calibrated After Initial Rank Ordering (CAIRO) framework. 


One caveat with Stage~1 optimization is that $g^{\star}$ is only identified up to a strictly increasing transformation. In classification tasks, this identification issue is absent because the scale of responses is fixed; however, in our regression framework, the Stage 2 optimizer $f^{\star}$ is consequently identified only up to the inverse of said transformation. Despite this individual ambiguity, the composition $f^{\star}(g^{\star}) = m^{\star}$ remains well-identified.  While one could fix the Stage~1 identification by requiring $g^{\star}(X)$ to follow a specific distribution (e.g., $\text{Unif}[0,1]$), this is often unnecessary in practice, since in the implementation of Stage~1 we will rely on smoothed versions of the rank loss.  For instance, in implementations utilizing the log-sigmoid function, $g^{\star}$ is identified up to an additive constant, allowing for a unique optimizer to be fixed by a simple constraint such as $g^{\star}(\mathbf{0})=0$.

\subsection{Stage 1: Optimal-in-Rank-Order Objectives}
\label{sec: ranking_phase}





In this section, we discuss Stage~1 of the decoupled regression framework (Stage~2 will be discussed in Section \ref{sec:calibration}). 
For two i.i.d. copies
$(X,Y)$ and $(X',Y')$ and a symmetric weighting function
$w(Y,Y')$, consider a generalization of the pairwise ranking loss in  \eqref{eq:01-pop}:
\begin{equation}
\label{weighted-01-pop}
\mathcal{L}
\;\equiv\;
\mathbb{E}\!\left[
  w(Y,Y')\,
  \mathbbm{1}\!\left\{
    (Y-Y')\,\operatorname{sign}\!\big(g(X)-g(X')\big) \le 0
  \right\}
\right].
\end{equation}
The pairwise structure of \eqref{weighted-01-pop} aligns with the fundamental objective of learning to rank: to construct a scoring function $g$ that recovers the relative order of the targets. 
Specifically, the loss uses an indicator function to penalize instances where the scores disagree with the true target order, thereby encouraging the condition  $g(X) > g(X') \;\Longrightarrow\; Y > Y'$ to hold with high probability.  Different choices of the weight function $w(Y,Y')$ leads to certain commonly used ranking losses: (i) the \emph{Uniform} weight $w(Y,Y') \equiv 1$ \citep{clemenccon2008ranking, burges2005learning};
(ii) the \emph{Absolute Gap} weight $w(Y,Y') \equiv |Y-Y'|$ \citep{agarwal2005stability, wang2018lambdaloss};
and (iii) the \emph{Rank Gap} weight $w(Y,Y') \equiv |F_Y(Y)-F_Y(Y')|$ where $F_Y$ stands for the distribution function of $Y$ \citep{menon2016bipartite}.
The introduction of these weights allows for a flexible ranking objective.

Additional insights can be obtained from Theorem~\ref{thm:equivalence} below, which provides the equivalent pointwise forms of loss \eqref{weighted-01-pop} with the specified weighting functions. This gives the theoretical basis for using the pointwise ranking loss in Stage~1 in \eqref{eq:stage1}. The proof of Theorem 1 is provided in Appendix \ref{proof:equivalence_of_losses}

\begin{theorem}
\label{thm:equivalence}
Let $\mathcal{L}$ be the loss defined in \eqref{weighted-01-pop}.  Then we have the following equivalent representations $\mathcal{L}_{\textup{uni}}$, $\mathcal{L}_{\textup{abs}}$, and $\mathcal{L}_{\textup{cdf}}$ based on three different choices of the weight $w(Y,Y')$:

\begin{enumerate}
    \item \textbf{Kendall's $\tau$:}
    \label{Kendall}
    If $w(Y,Y') = 1$, minimizing $\mathcal{L}$ with respect to $g$ is equivalent to minimizing $\mathcal{L}_{\textup{uni}}$ involving Kendall's rank correlation coefficient $\tau \equiv \mathbb{E}[\operatorname{sign}(Y-Y')\,\operatorname{sign}(g(X)-g(X'))]$ between $Y$ and $g(X)$:
    \begin{equation}
        \mathcal{L}_{\textup{uni}}(Y, g(X)) = \frac{1}{2} ( 1 - \tau(Y, g(X)) ).
    \end{equation}

    \item \textbf{Gini covariance:}
    \label{Gini}
    If $w(Y,Y') = |Y - Y'|$, minimizing $\mathcal{L}$ with respect to $g$ is equivalent to minimizing $\mathcal{L}_{\textup{abs}}$ involving the Gini covariance between $Y$ and the rank of $g(X)$:
    \begin{equation}
        \mathcal{L}_{\textup{abs}}(Y, g(X)) = c - 2\operatorname{Cov}(Y, F_{g(X)}(g(X))).
    \label{eq: loss_abs}
    \end{equation}
    Here $c=\frac{1}{2}\mathbb{E}|Y-Y'|$, which does not affect optimization of $\mathcal{L}_{\textup{abs}}$ with respect to $g$. 
    \item \textbf{Spearman's $\rho$:}
    \label{Spearman}
    If $w(Y,Y') = |F_Y(Y) - F_Y(Y')|$, minimizing $\mathcal{L}$ with respect to $g$ is equivalent to minimizing $\mathcal{L}_{\textup{cdf}}$ involving Spearman's rank correlation $\rho_S \equiv 12\cdot\mathrm{Cov}(F_Y(Y),F_g(g(X)))$ between $Y$ and $g(X)$:
    \begin{equation}
        \mathcal{L}_{\textup{cdf}}(Y, g(X)) = \frac{1}{6} ( 1 - \rho_S(Y, g(X)) ).
    \end{equation}
\end{enumerate}


\end{theorem}
Theorem~\ref{thm:equivalence} provides substantial interpretability for these weighted ranking losses by connecting them to classical statistical measures. (i) For $w=1$, as established in Theorem~\ref{thm:equivalence}, this loss acts as a direct proxy for Kendall's $\tau$. Thus, penalizing every misordered pair equally is mathematically equivalent to maximizing the probability of concordance between the target $Y$ and the model score $g(X)$. (ii) For $w=|Y-Y'|$, Theorem~\ref{thm:equivalence} reveals that this loss is the negative of the Gini covariance (up to a constant). It follows that assigning higher penalties to pairs distant in the target space is equivalent to maximizing the association between the response $Y$ and the score's cumulative distribution $F_{g(X)}(g(X))$. The Gini covariance is a standard tool in actuarial science for quantifying the relationship between risk scores and insurance losses and is extensively studied by \citet{yitzhaki2012more}. (iii) For $w = |F_Y(Y) - F_Y(Y')|$, Theorem~\ref{thm:equivalence} confirms that the loss is a linear transformation of Spearman's rank correlation $\rho_S$: hence, prioritizing pairs that are far apart in quantile space is equivalent to maximizing the correlation between the ranks of $Y$ and $g(X)$. This variant offers a compromise; it captures the relative distribution of the targets while remaining invariant to extreme outliers in $Y$ that might otherwise destabilize the Gini objective.

We now characterize the family of ranking losses suitable for the first stage of the CAIRO framework. We focus specifically on the subclass of losses that recover the correct ordering of the conditional mean as their optimal solution. To formalize this requirement, we define the following property:

\begin{definition}
\label{def:ORO}
  A loss function $\mathcal{L}$ is Optimal in Rank Order (ORO) if any of its associated population minimizer $g^{\star}$ from \eqref{eq:stage1} belongs to $\mathcal{G}'$ which is defined as:
  \[
    \mathcal{G}'\equiv\big\{h( m^{\star}):\ h:\mathbb{R}\to\mathbb{R}\text{ strictly increasing}\big\} \subset \mathcal{G}.
  \]
\end{definition}
This definition is central to the CAIRO framework because the success of calibration in Stage~2 rests entirely on the ordering established in the ranking stage. Since calibration applies monotonic transformation, it can only adjust the scale of predictions from $g^\star$ while strictly preserving their order. Consequently, if the Stage~1 objective fails to satisfy the ORO property, the resulting score $g^\star$ will not align with the ranking of the conditional mean $m^\star$; in such cases, no amount of recalibration can recover the true regression function $m^\star$. By restricting our attention to ORO losses, we ensure that the decoupled framework remains theoretically consistent with standard regression objectives.

The following theorem shows that the three losses defined in Theorem~\ref{thm:equivalence} achieve the ORO property under specific conditions: 

\begin{theorem}
\label{thm:oro_all_three}
Let $(X,Y)$ and $(X',Y')$ be i.i.d. The weighted pairwise ranking losses defined in Theorem~\ref{thm:equivalence} are ORO under the following conditions:

\begin{enumerate}
    
    \item \textbf{Kendall's $\tau$:} 
    $\mathcal{L}_{\textup{uni}}$  is ORO if the distribution of $Y$ is pairwise consistent with  $m^\star$:
    \[
    m^\star(x)>m^\star(x') \quad\Longleftrightarrow\quad \mathbb{P}(Y>Y'|X=x,\,X'=x') > 1/2 .
    \]

    \item \textbf{Gini covariance:} 
    Provided $\mathbb{E}[|Y|] < \infty$, $\mathcal{L}_{\textup{abs}}$ is ORO. 

    \item \textbf{Spearman's $\rho$:} 
    $\mathcal{L}_{\textup{cdf}}$  is ORO if the Additive-Noise Model holds:
    \[
    Y=m^\star(X)+\varepsilon,
    \]
    where $\varepsilon \perp X$ and $\varepsilon$ has a continuous, strictly increasing CDF.
\end{enumerate}
\end{theorem}

\begin{remark}
    The results in Theorem~\ref{thm:oro_all_three} illustrate how different ranking objectives rely on varying structural assumptions to achieve the ORO property. (i) As can be seen in the proof of Theorem \ref{thm:oro_all_three} in Appendix \ref{proof_of_oro}, the Gini covariance objective is governed directly by the difference in conditional means $m^\star(x) - m^\star(x')$; consequently, it satisfies the ORO property under very mild assumptions. (ii) The Kendall objective targets the probability of concordance, requiring only that the likelihood of $Y > Y'$ remains consistent with the ordering of the means. (iii) In contrast, the Spearman objective operates in quantile space; because the ordering of expected quantiles does not always coincide with the ordering of expectations, this objective requires the stronger structural assumption of additive noise to guarantee the ORO property.
\end{remark}

\subsection{Stage 2: Calibration via Isotonic Regression}
\label{sec:calibration}

In Stage~2 we restore the scale of the regression function.  
Recall that $m^\star(x)=\mathbb{E}[Y|X=x]$ and let $g^\star$ denote Stage~1 score.  
Under any ORO loss, an optimal population ranker satisfies $g^\star(x) = h( m^\star(x))$ for some strictly increasing $h:\mathbb R\to\mathbb R.$ The goal of Stage~2 is to learn a monotone map that inverts the unknown transform $h$ and recovers $m^\star$.


Specifically, the Stage~2 calibration problem seeks a nondecreasing function that best aligns these scores with the outcome $Y$.  
Let $\mathcal F \equiv \{ f : \mathbb{R} \to \mathbb{R} \ \text{nondecreasing} \}.$ We define the population version of the objective:
\begin{equation}
\label{eq:pop_iso_again}
f^\star \in \arg\min_{f\in\mathcal F} \ \mathbb{E}\big[(Y - f(g^\star(X)))^2\big].
\end{equation}
When composed with $g^\star$, the optimizer $f^\star$ of \eqref{eq:pop_iso_again} recovers $m^\star$, so calibration recovers the correct regression function. The following theorem formalizes this statement; the proof is given in Appendix~\ref{composition_is_f_star}.

\begin{theorem}
\label{composition_is_f_star_statement}
Assume $\mathbb{E}[Y^2] < \infty$ and suppose that the Stage~1 score satisfies $g^{\star}(x) = h(m^\star(x))$ for some strictly increasing $h:\mathbb R\to\mathbb R$. In Stage 2, let $f^{\star}$ be any minimizer of the monotone least-squares problem in \eqref{eq:pop_iso_again}.  
Then it recovers the true regression function
\[
f^{\star}\big(g^{\star}(X)\big) = m^\star(X) .
\]
\end{theorem}

\begin{remark}[Connection to auto-calibration and proper scoring]
\label{rmk:auto_calibration}
A prediction function $m(X)$ is strictly auto-calibrated if it satisfies $m(X) = \mathbb{E}[Y | m(X)]$ \citep{kruger2021generic, gneiting2023regression}. The true conditional mean $m^\star(X) = \mathbb{E}[Y|X]$ naturally satisfies this property:
$
\mathbb{E}[Y | m^\star(X)] = \mathbb{E}[\mathbb{E}[Y | X] \mid m^\star(X)] = m^\star(X).
$
While the Stage~1 ranking objectives (such as the Gini covariance) excel at learning robust, scale-invariant representations, they are fundamentally improper scoring rules when evaluated over unconstrained function classes; their raw outputs $g^\star(X)$ do not correspond to the true conditional expectation and are therefore not auto-calibrated.  Stage~2 resolves this  limitation by solving the monotone least-squares problem, isotonic regression explicitly estimates the conditional expectation $\mathbb{E}[Y | g^\star(X)]$ while preserving the order learned in Stage~1. As established in Theorem~\ref{composition_is_f_star_statement}, composing the uncalibrated ranker with this monotonic projection recovers $m^\star(X)$. 

Consequently, \textsc{CAIRO} mathematically guarantees that the final predictions are auto-calibrated, effectively converting an improper ranking mechanism into a statistically rigorous regression estimator. Crucially, because isotonic regression estimates conditional expectations via local block-averaging, this auto-calibration property holds not just asymptotically, but strictly on the empirical training data, even when the Stage~1 ranker is estimated with finite-sample noise.
\end{remark}



\section{Estimation and Algorithms}

We implement the \textsc{CAIRO} framework by approximating the population losses from Section \ref{sec:our_method} using an i.i.d. sample $\{(x_i, y_i)\}_{i=1}^n$. Since these objectives are non-differentiable, we derive empirical analogs in both pairwise and pointwise forms and apply smooth surrogates to enable gradient-based optimization. The second stage solves the empirical isotonic regression problem. This section formalizes these smoothing strategies and analyzes the resulting computational complexity.

\paragraph{Pairwise empirical losses (Stage 1).} 
We first consider the empirical version of the pairwise losses for Stage~1. Assuming no ties in the targets $y_i$'s, and symmetric weights $w_{ij}=w_{ji}\ge 0$, we define the empirical version of \eqref{weighted-01-pop}:
\begin{equation}
\label{eq:emp-oroloss}
    \ell(y, g(x))
    \equiv
    \frac{1}{n(n-1)}\sum_{1\le i<j\le n}
    w_{ij}\,
    \mathbbm{1}\!\left\{(y_i-y_j)\big(g(x_i)-g(x_j)\big)\le 0\right\}.
\end{equation}
This template unifies our Stage~1 objectives at the empirical level through the choice of $w_{ij}$. Specifically, we utilize three weights: (i) the uniform weights $w_{\text{uni}} \equiv 1$, (ii) the absolute gap weights $w_{\text{abs}} \equiv |y_i - y_j|$, and (iii) the rank gap weights $w_{\text{cdf}} \equiv |\hat{F}_n(y_i) - \hat{F}_n(y_j)|$, where $\hat{F}_n$ is the empirical CDF of $Y$. For computational purposes, we consider an equivalent form of \eqref{eq:emp-oroloss}  by isolating the misordered pairs (see proof in Appendix~\ref{app:emp-oroloss-eq})
\begin{equation}
\label{eq:emp-oroloss-ordered}
    \ell(y, g(x))
    \;=\;
    \frac{1}{n(n-1)}\sum_{i\ne j}
    w_{ij}\,
    \mathbbm{1}\{y_i>y_j\}\,\mathbbm{1}\{g(x_i)<g(x_j)\},
\end{equation}
The hard indicator function in \eqref{eq:emp-oroloss-ordered} results in a non-differentiable objective, so we replace it with a smooth surrogate. Specifically, observe that the indicator function respects the following upper bound:


\begin{equation}
\label{eq:weighted-logistic-bound}
\mathbbm{1}\{g(x_i)<g(x_j)\}
\;\le\;
\log\!\left(1+e^{-\sigma\,(g(x_i)-g(x_j))}\right) .
\end{equation}
We refer to the right hand side of \eqref{eq:weighted-logistic-bound} as the log-sigmoid function (that is evaluated at $g(x_i)-g(x_j)$), where $\sigma$ is a hyperparameter that characterizes the shape of the sigmoid. The log-sigmoid function has been extensively used as a smooth, convex surrogate for the indicator function in recent literature since an upper bound on \eqref{eq:emp-oroloss-ordered} can be derived using \eqref{eq:weighted-logistic-bound}: 
\begin{align}
\label{eq:weighted-ranknet}
\ell(y, g(x))
\;\le\;
\frac{1}{n(n-1)}\sum_{i\ne j}
w_{ij}\mathbbm{1}\{y_i>y_j\}\,
\log\!\left(1+e^{-\sigma\,(g(x_i)-g(x_j))}\right).
\end{align}
In practice, we would use the right hand side of \eqref{eq:weighted-ranknet} as the loss function in the implementation. As an upper bound it naturally decreases the true empirical loss function \eqref{eq:emp-oroloss-ordered}.

We can see that the RankNet loss is \eqref{eq:weighted-ranknet} by setting $w_{ij}= w_{\textup{uni}} = 1$. This result shows that the RankNet loss can be viewed either as   cross-entropy between the true probability and a sigmoid-modeled posterior probability \citep{burges2005learning}, or as a convex upper bound on the pairwise error,  corresponding to two alternative interpretations of RankNet:  either a probabilistic pairwise classifier or a smooth surrogate for improving Kendall's $\tau$, respectively. This fact supports our ORO property by showing that it includes standard, widely used ranking losses.

\paragraph{Pointwise empirical losses (Stage 1).}  By Theorem~\ref{thm:equivalence}, the Stage~1 losses based on the Gini covariance and Spearman's $\rho$ also have their pointwise versions in terms of the relevant covariances.  Now we consider the empirical version of these pointwise losses.  We construct these losses by replacing the population CDF with its sample counterpart, which can then be written using the \emph{rank} function in \eqref{eq:emp_cdf_is_rank}. Because the \emph{rank} operator is piecewise constant with vanishing gradients, we employ a differentiable \emph{softrank} approximation \citep{blondel2020fast, cuturi2019differentiable} to enable gradient-based optimization.
We illustrate this approach with the Gini covariance loss $\mathcal{L}_{\textup{abs}}$. Consider the empirical version: $\ell_{\textup{abs}}(y, g(x)) = -\,{\widehat{\mathrm{Cov}}}(y_i,\widehat F_{g(X)}(g(x_i))\big)$. We show in Appendix \ref{gini_form} that  $\ell_{\textup{abs}}$ may be written equivalently, up to additive and positive multiplicative (scalar) constants, as:
\begin{equation}
\ell_{\textup{abs}}(y, g(x))
\;\equiv\;
-\sum_{i=1}^n y_i\,\operatorname{rank}\!\big(g(x_i)\big).
\label{eq:rank_form_of_gini_no_constants}
\end{equation}

We can smooth \eqref{eq:rank_form_of_gini_no_constants} by replacing the \emph{rank} function with \emph{softrank} \citep{blondel2020fast, cuturi2019differentiable}, yielding the smooth objective  $\ell_\textup{soft-abs}$ defined as:
\begin{equation}
\label{eq: smooth_gini_loss}
        \ell_\textup{soft-abs}(y, g(x))
    \equiv\;-\sum_{i=1}^n y_i\,\operatorname{softrank}\!\big(g(x_i)\big), 
\end{equation}

\begin{remark*}
    \eqref{eq:rank_form_of_gini_no_constants} has been studied by \citet{frees2011summarizing} and \citet{yitzhaki2012more} as a metric for evaluating insurance risk scores and measures of association. \citet{wang2018lambdaloss} used \eqref{eq:rank_form_of_gini_no_constants} to derive \eqref{eq:emp-oroloss-ordered} with $w_{ij} = |y_i-y_j|$ which they implemented with log-sigmoid smoothing. To our knowledge, this work is the first to analyze \eqref{eq: smooth_gini_loss} as a differentiable pointwise loss function implemented via smooth rank approximations.
\end{remark*}

\label{sec:calib-empirical}

\paragraph{Empirical Calibration (Stage 2).}
Having learned a scoring rule $\hat{g}(x_i)$ from a Stage~1 loss for data $\{(x_i,y_i)\}_{i=1}^n$, we now restore the target scale via monotone calibration by solving the empirical analogue of~\eqref{eq:pop_iso_again} to the dataset $\{(\hat{g}(x_i),y_i)\}_{i=1}^n$:
\begin{equation}
\label{eq:emp_iso}
\min_{f \in \mathcal{F}}\ \frac{1}{n}\sum_{i=1}^n \big(y_i- f(\hat{g}(x_i))\big)^2
\quad \text{s.t.}\quad f(\hat{g}(x_i))\le  f(\hat{g}(x_j))\ \text{whenever } \hat{g}(x_i)\le \hat{g}(x_j).
\end{equation}
Let $\tilde{f}$ denote the resulting solution to \eqref{eq:emp_iso}. Because the Pool Adjacent Violators (PAV) algorithm fits piecewise-constant steps by taking the exact empirical average of the responses $y_i$ within each block, the resulting predictor $\hat{m}(x_i) = \tilde{f}(\hat{g}(x_i))$ is guaranteed to be strictly empirically auto-calibrated on the training set. This provides a robust statistical safety net: even if the finite-sample estimate $\hat{g}$ deviates from the true optimal ordering, the final predictions remain locally unbiased with respect to the learned representation. To further arrive at a continuous mapping $\hat{f}$ for out-of-sample prediction, we utilize the implementation from \citet{pedregosa2011scikit}, which linearly interpolates between the original fitted values, the $\tilde{f}(\hat{g}(x_i))$'s, and clips predictions outside the training range.

\paragraph{Computational considerations.}
The computational complexity of \textsc{CAIRO} is governed by the choice of empirical loss in Stage~1. When optimizing pairwise objectives \eqref{eq:weighted-ranknet}, we employ log-sigmoid smoothing, which requires computing differences across all pairs in a batch, yielding $O(n^2)$ complexity. Conversely, for pointwise objectives \eqref{eq:rank_form_of_gini_no_constants}, we utilize the differentiable \textit{softrank} operator, which \citet{blondel2020fast} demonstrated can be computed in $O(n \log n)$ time. Since Stage~2 calibration via the Pool Adjacent Violators (PAV) algorithm requires only linear time on sorted data \citep{barlow1972statistical}, the overall training complexity scales as $O(n^2)$ for pairwise losses and $O(n \log n)$ for pointwise losses.

\section{Numerical Studies}
\label{sec:experiments}

We empirically validate the \textsc{CAIRO} framework through controlled simulations and real-world regression benchmarks. Our primary objective is to assess regression and ranking accuracy  across varying noise regimes, isolating the impact of the Stage~1 ranking objective.

\paragraph{CAIRO model variants.}
We implement three model variants of the \textsc{CAIRO} framework. All variants share the same Stage~2 procedure (isotonic calibration) but differ in their Stage~1 loss and smoothing strategy:
\begin{itemize}
    \item \texttt{RankNet}: Pairwise RankNet loss \eqref{eq:weighted-ranknet} with uniform weights $w_{ij}=1$ and log-sigmoid smoothing.
    \item \texttt{RankNet-GiniW}: Pairwise RankNet loss \eqref{eq:weighted-ranknet} with absolute value weights $w_{ij}=|y_i-y_j|$ and log-sigmoid smoothing.
    \item \texttt{GiniNet-SoftRank}: Pointwise Gini covariance loss \eqref{eq: smooth_gini_loss} with \textit{softrank} smoothing.
\end{itemize}
Stage~1 scores are parameterized by a two-layer MLP ($32 \times 16$, ReLU), optimized via Adam. Figure~\ref{fig:mechanism_comparison} illustrates this two-stage learning process for the \texttt{RankNet} model variant across data regimes~1 and 3.

\paragraph{Baselines and metrics.}
We benchmark against three regression standards: (i) \texttt{NN-MSE}, an MLP with the identical architecture trained on squared error; (ii) \texttt{RandomForest} (RF), an ensemble of CART trees; and (iii) \texttt{LightGBM}, a gradient-boosted decision tree model minimizing MSE. Performance is evaluated using Root Mean Squared Error (RMSE) for regression accuracy, and Spearman's $\rho_S$ and Kendall's $\tau$ for ranking quality.

\subsection{Synthetic Data-Generating Processes}
\label{subsec:dgps}
Covariates $X \in \mathbb{R}^{d}$ and weights $w$ are drawn from $\mathcal{N}(0, I_d)$, with $n=6000$ and $d=10$. We define the linear function $\eta = X^\top w / \sqrt{d}$ and generate responses $Y$ by applying three distinct mechanisms to this linear function $\eta$:

\begin{enumerate}
    \item \textit{Normal:} A standard linear regression setting where $Y = \eta + \varepsilon$ with $\varepsilon \sim \mathcal{N}(0,1)$. This serves as a baseline where MSE is theoretically optimal.
    
    \item \textit{Heteroskedastic and Gamma tail}: We apply a canonical log-link, setting $\mu = \exp(\eta)$. Responses are drawn from $Y \sim \mathrm{Gamma}(k, \mu/k)$ with shape $k=2$. Here, the variance $\mu/k$ scales with the mean $\mu$, introducing heteroskedasticity. 
    
    \item \textit{Heteroskedastic and heavy tail}: We again use the mean $\mu = \exp(\eta)$ but introduce additive, heteroskedastic and heavy-tailed noise: $Y = \mu + \varepsilon \cdot \sqrt{\mu}$, where $\varepsilon$ is drawn from a Lognormal distribution.
\end{enumerate}

\subsection{Results on Synthetic Data}
\label{subsec:synthetic_results}

\begin{figure*} 
    \centering
\includegraphics[width=\textwidth]{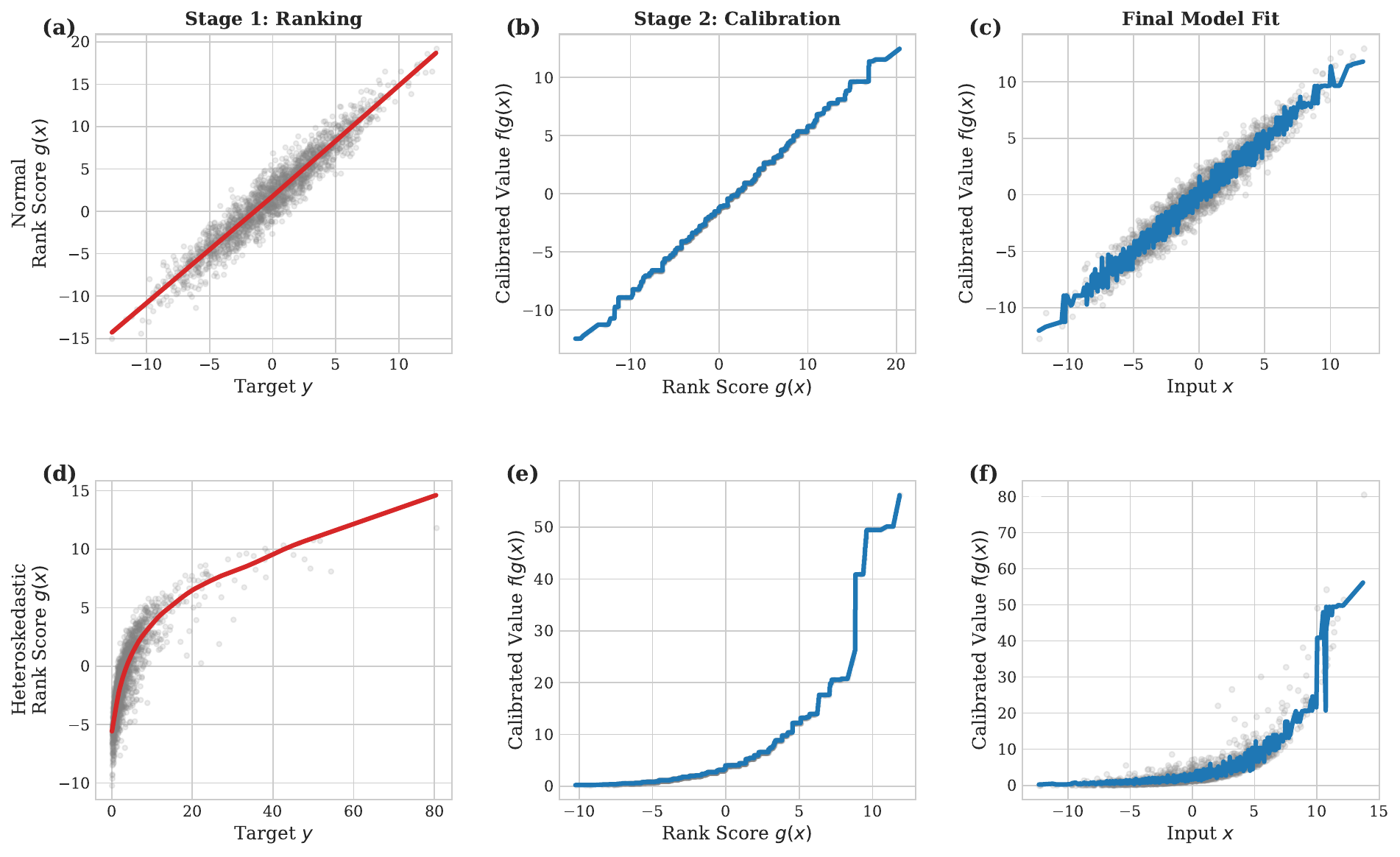}
    \caption{\texttt{CAIRO-RankNet} across Data Regimes. 
    \textbf{Top Row} (Normal Regime): Stage~1 (a) learns a linear ranking, and Stage~2 (b) learns a linear calibration map, resulting in a standard regression fit (c).
    \textbf{Bottom Row} (Heteroskedastic and Heavy Tail Regime): In the presence of heavy-tailed noise, Stage~1 (d) learns a non-linear ranking to preserve order. Stage~2 (e) learns a non-parametric step function to correct this warp. The final model (f) successfully recovers the underlying signal (blue line) despite the large variance in the raw data (gray cloud).}
    \label{fig:mechanism_comparison}
\end{figure*}

The results for the synthetic experiments are reported in Table~\ref{tab:clean_all}.  We provide some comments below.

\paragraph{Performance under homoskedastic noise.}
In the \textit{Normal} setting, the \texttt{NN-MSE} baseline is theoretically optimal; it achieves the lowest RMSE as expected. \texttt{GiniNet-SoftRank} proves highly effective here, matching the baseline's precision while maintaining strong ranking metrics. \texttt{RankNet-GiniW} underperforms slightly, suggesting that pointwise objectives (via SoftRank) are better suited in high-SNR regimes.

\paragraph{Robustness to heavy-tailed noise.}
The advantages of the \textsc{CAIRO} framework emerge in the \textit{Heteroskedastic Tail} setting, especially under heavy-tailed noise. Here, standard MSE minimization fails, as the squared error loss is disproportionately sensitive to extreme outliers. In contrast, all \textsc{CAIRO} variants remain stable. Notably, \texttt{RankNet(-GiniW)} achieves the best performance, outperforming both the baseline and the \texttt{GiniNet-SoftRank}. We attribute this robustness to the log-sigmoid component of the pairwise loss, which saturates for large error magnitudes, effectively capping the gradient contribution of outliers. While \texttt{GiniNet-SoftRank} still improves upon the MSE baseline, it does not match the robustness of the pairwise approach in this regime.

\begin{table}[H]
    \centering
    \caption{Synthetic data performance across noise models: mean (95\% CI) over $5$ repetitions. Best value per scenario/metric is in bold.}
    \label{tab:clean_all}
    \renewcommand{\arraystretch}{1.05} 
    \begin{tabular}{llccc}
        \toprule
         & Model & Spearman's $\rho$ & Kendall's $\tau$ & RMSE \\
        \midrule
        Scenario & \multicolumn{4}{c}{Normal ($\sigma = 1.0$)} \\
        \midrule
        & CAIRO--RankNet
        & $0.594 \pm 0.074$ & $0.435 \pm 0.059$ & $1.093 \pm 0.026$ \\
        & CAIRO--RankNet-GiniW
        & $0.588 \pm 0.074$ & $0.431 \pm 0.056$ & $1.095 \pm 0.025$ \\
        & CAIRO--GiniNet-SoftRank
        & $0.654 \pm 0.051$ & $\mathbf{0.484 \pm 0.041}$ & $1.031 \pm 0.010$ \\
        & NN-MSE
        & $\mathbf{0.667 \pm 0.045}$ & $0.481 \pm 0.037$ & $\mathbf{1.004 \pm 0.008}$ \\
        & Random Forest
        & $0.642 \pm 0.045$ & $0.459 \pm 0.037$ & $1.039 \pm 0.020$ \\
        & LightGBM--MSE
        & $0.643 \pm 0.056$ & $0.461 \pm 0.046$ & $1.033 \pm 0.019$ \\
        \midrule
        \multicolumn{5}{c}{Heteroskedastic and Gamma tail (shape $k = 2$)} \\
        \midrule
        & CAIRO--RankNet
        & $0.701 \pm 0.058$ & $0.527 \pm 0.050$ & $2.238 \pm 0.726$ \\
        & CAIRO--RankNet-GiniW
        & $0.714 \pm 0.058$ & $0.540 \pm 0.050$ & $2.082 \pm 0.577$ \\
        & CAIRO--GiniNet-SoftRank
        & $0.730 \pm 0.040$ & $\mathbf{0.553 \pm 0.037}$ & $\mathbf{2.040 \pm 0.554}$ \\
        & NN-MSE
        & $\mathbf{0.736 \pm 0.040}$ & $0.544 \pm 0.037$ & $2.125 \pm 0.775$ \\
        & Random Forest
        & $0.698 \pm 0.034$ & $0.510 \pm 0.030$ & $2.176 \pm 0.783$ \\
        & LightGBM--MSE
        & $0.695 \pm 0.037$ & $0.508 \pm 0.032$ & $2.112 \pm 0.687$ \\
        \midrule
        \multicolumn{5}{c}{Heteroskedastic and heavy tail} \\
        \midrule
        & CAIRO--RankNet
        & $0.858 \pm 0.037$ & $0.693 \pm 0.043$ & $\mathbf{1.423 \pm 0.089}$ \\
        & CAIRO--RankNet-GiniW
        & $0.860 \pm 0.039$ & $0.696 \pm 0.044$ & $1.449 \pm 0.095$ \\
        & CAIRO--GiniNet-SoftRank
        & $\mathbf{0.863 \pm 0.033}$ & $\mathbf{0.698 \pm 0.039}$ & $1.580 \pm 0.239$ \\
        & NN-MSE
        & $0.848 \pm 0.044$ & $0.667 \pm 0.047$ & $1.882 \pm 0.513$ \\
        & Random Forest
        & $0.830 \pm 0.026$ & $0.642 \pm 0.028$ & $1.873 \pm 0.533$ \\
        & LightGBM--MSE
        & $0.854 \pm 0.026$ & $0.672 \pm 0.030$ & $1.705 \pm 0.470$ \\
        \bottomrule
    \end{tabular}
\end{table}

\subsection{Real-Data Analysis}
\label{subsec:realdata_setup}

To assess generalization, we evaluate our method on five regression benchmarks from the UCI Machine Learning Repository \citep{frank2010uci}: Auto MPG (AUTO), Communities and Crime (CCRI), Computer Hardware (COMP), Concrete Slump Test (CSLU), and Abalone (ABALONE). These datasets vary in size ($n \in [103, 4177]$) and dimensionality ($p \in [6, 102]$); see Table~\ref{tab:realdata}.

\paragraph{Preprocessing.}
We follow the preprocessing protocol of \citet{roy2012robustness}. For AUTO, we remove the high-cardinality \emph{car name} feature and incomplete cases (resulting in $n=392$). For CCRI, we discard predictors with more than $58\%$ missing values (resulting in $p=102$). For ABALONE, we regress strictly on continuous shell measurements to predict Age.  All datasets are devoid of missing values after preprocessing.

\paragraph{Results.}
We compare \textsc{CAIRO} variants against the regression baselines detailed in the introductory paragraphs of Section~\ref{sec:experiments}. Table~\ref{tab:real_all} reports performance across five random splits.

First, we observe that \texttt{NN-MSE} exhibits significant instability, particularly on smaller datasets like AUTO and CSLU. In contrast, the \textsc{CAIRO} variants, specifically \texttt{RankNet-GiniW} and \texttt{GiniNet-SoftRank}, demonstrate much better results using identical neural network architectures. This suggests that the rank-based objectives in CAIRO variants may stabilize neural networks in low-data regimes.

Second, \textsc{CAIRO} proves competitive with tree ensembles, which are typically dominant in tabular domains. On the \textit{Computer Hardware} (COMP) dataset, characterized by significant price outliers, \texttt{RankNet-GiniW} achieves an RMSE of $73.7$, substantially outperforming Random Forest ($87.8$) and LightGBM ($94.2$); this mirrors our synthetic results regarding the robustness of pairwise losses. On \textit{Concrete Slump} (CSLU), the smallest dataset ($n=103$), \textsc{CAIRO} outperforms all baselines (RMSE $3.52$ vs. $4.45$ for RF), indicating superior data efficiency. Finally, on the larger ABALONE task, \textsc{CAIRO} yields the highest ranking correlations ($\rho \approx 0.78$) and lowest RMSE ($2.12$). These results demonstrate that \textsc{CAIRO} effectively closes the performance gap between neural networks and tree-based models on tabular regression tasks.


\begin{table}[H]
    \centering
    \caption{Characteristics of the real datasets used in our study.}
    \label{tab:realdata}
    \renewcommand{\arraystretch}{1.05} 
    \begin{tabular}{lcccc}
        \toprule
        Dataset & No.\ of data points ($n$) & No.\ of predictors ($p$) & No.\ of unique response values & Source \\
        \midrule
        AUTO     & 392  & 7   & 127 & UCI \\
        CCRI     & 1994 & 102 &  98 & UCI \\
        COMP     & 209  & 6   & 116 & UCI \\
        CSLU     & 103  & 7   &  90 & UCI \\
        ABALONE  & 4177 & 7   &  29 & UCI \\
        \bottomrule
    \end{tabular}
\end{table}

\begin{table} 
    \centering
    \caption{Real data performance on five UCI regression tasks.
    Entries are mean (95\% CI) over five random 70/30 train--test splits.
    Best value per dataset/metric is in bold.}
    \label{tab:real_all}
    \renewcommand{\arraystretch}{1} 
    \begin{tabular}{llccc}
        \toprule
         & Model & Spearman's $\rho$ & Kendall's $\tau$ & RMSE \\
        \midrule
        Dataset & \multicolumn{4}{c}{AUTO (Auto MPG)} \\
        \midrule
        & \textsc{CAIRO}--GiniNet-SoftRank
        & $0.870 \pm 0.022$ & $0.719 \pm 0.022$ & $4.244 \pm 0.294$ \\
        & \textsc{CAIRO}--RankNet
        & $0.933 \pm 0.011$ & $0.795 \pm 0.020$ & $3.188 \pm 0.204$ \\
        & \textsc{CAIRO}--RankNet-GiniW
        & $0.923 \pm 0.025$ & $0.787 \pm 0.027$ & $3.405 \pm 0.455$ \\
        & NN-MSE
        & $0.050 \pm 0.578$ & $0.026 \pm 0.434$ & $12.321 \pm 4.540$ \\
        & Random Forest
        & $\mathbf{0.937 \pm 0.015}$ & $\mathbf{0.798 \pm 0.017}$ & $\mathbf{2.994 \pm 0.239}$ \\
        & LightGBM--MSE
        & $0.935 \pm 0.015$ & $0.796 \pm 0.019$ & $3.010 \pm 0.230$ \\
        \midrule
        \multicolumn{5}{c}{CCRI (Communities and Crime)} \\
        \midrule
        & \textsc{CAIRO}--GiniNet-SoftRank
        & $0.810 \pm 0.018$ & $\mathbf{0.640 \pm 0.020}$ & $0.142 \pm 0.002$ \\
        & \textsc{CAIRO}--RankNet
        & $0.757 \pm 0.013$ & $0.586 \pm 0.013$ & $0.160 \pm 0.005$ \\
        & \textsc{CAIRO}--RankNet-GiniW
        & $0.762 \pm 0.010$ & $0.593 \pm 0.011$ & $0.159 \pm 0.004$ \\
        & NN-MSE
        & $0.475 \pm 0.371$ & $0.358 \pm 0.280$ & $0.190 \pm 0.053$ \\
        & Random Forest
        & $\mathbf{0.817 \pm 0.015}$ & $0.631 \pm 0.016$ & $\mathbf{0.138 \pm 0.004}$ \\
        & LightGBM--MSE
        & $0.812 \pm 0.012$ & $0.626 \pm 0.013$ & $\mathbf{0.138 \pm 0.003}$ \\
        \midrule
        \multicolumn{5}{c}{COMP (Computer Hardware)} \\
        \midrule
        & \textsc{CAIRO}--GiniNet-SoftRank
        & $0.784 \pm 0.042$ & $0.644 \pm 0.042$ & $105.055 \pm 64.484$ \\
        & \textsc{CAIRO}--RankNet
        & $0.800 \pm 0.072$ & $0.657 \pm 0.070$ & $89.822 \pm 31.436$ \\
        & \textsc{CAIRO}--RankNet-GiniW
        & $0.833 \pm 0.043$ & $0.688 \pm 0.045$ & $\mathbf{73.741 \pm 34.310}$ \\
        & NN-MSE
        & $0.392 \pm 0.363$ & $0.297 \pm 0.294$ & $334.020 \pm 142.277$ \\
        & Random Forest
        & $\mathbf{0.904 \pm 0.022}$ & $\mathbf{0.753 \pm 0.031}$ & $87.826 \pm 27.136$ \\
        & LightGBM--MSE
        & $0.808 \pm 0.077$ & $0.638 \pm 0.079$ & $94.157 \pm 27.191$ \\
        \midrule
        \multicolumn{5}{c}{CSLU (Concrete Slump Test)} \\
        \midrule
        & \textsc{CAIRO}--GiniNet-SoftRank
        & $0.769 \pm 0.114$ & $0.604 \pm 0.106$ & $5.207 \pm 0.895$ \\
        & \textsc{CAIRO}--RankNet
        & $0.822 \pm 0.189$ & $0.717 \pm 0.145$ & $4.807 \pm 2.306$ \\
        & \textsc{CAIRO}--RankNet-GiniW
        & $\mathbf{0.900 \pm 0.013}$ & $\mathbf{0.760 \pm 0.027}$ & $\mathbf{3.522 \pm 0.420}$ \\
        & NN-MSE
        & $0.033 \pm 0.409$ & $0.038 \pm 0.320$ & $23.394 \pm 12.779$ \\
        & Random Forest
        & $0.890 \pm 0.030$ & $0.733 \pm 0.043$ & $4.445 \pm 0.394$ \\
        & LightGBM--MSE
        & $0.844 \pm 0.034$ & $0.682 \pm 0.042$ & $4.519 \pm 0.251$ \\
        \midrule
        \multicolumn{5}{c}{ABALONE} \\
        \midrule
        & \textsc{CAIRO}--GiniNet-SoftRank
        & $0.631 \pm 0.040$ & $0.503 \pm 0.034$ & $2.610 \pm 0.042$ \\
        & \textsc{CAIRO}--RankNet
        & $\mathbf{0.781 \pm 0.014}$ & $\mathbf{0.637 \pm 0.013}$ & $\mathbf{2.115 \pm 0.042}$ \\
        & \textsc{CAIRO}--RankNet-GiniW
        & $0.778 \pm 0.015$ & $0.635 \pm 0.013$ & $2.116 \pm 0.036$ \\
        & NN-MSE
        & $0.634 \pm 0.028$ & $0.481 \pm 0.023$ & $3.255 \pm 0.102$ \\
        & Random Forest
        & $0.763 \pm 0.019$ & $0.606 \pm 0.017$ & $2.172 \pm 0.023$ \\
        & LightGBM--MSE
        & $0.761 \pm 0.015$ & $0.603 \pm 0.014$ & $2.192 \pm 0.045$ \\
        \bottomrule
    \end{tabular}
\end{table}

\section{Discussions}

We have introduced \textsc{CAIRO}, a framework that decouples regression into two distinct statistical tasks: learning the correct ordering of predictions and subsequently recovering their scale. By prioritizing the ranking objective, we circumvent the sensitivity of standard pointwise losses, such as Mean Squared Error (MSE), to outliers and heavy-tailed noise. Our theoretical analysis identifies a broad class of Optimal-in-Rank-Order (ORO) loss functions---including smooth approximations of Kendall's $\tau$ and Gini covariance---that guarantee consistency with the true conditional mean. We further prove that composing these rankers with isotonic calibration recovers the regression function $m^*(x) = \mathbb{E}[Y|X=x]$ at the population level. 

The decoupling of order and scale provides two statistical benefits: robustness and auto-calibration. First, as demonstrated in our noise experiments, standard MSE minimization is often driven by a few high-leverage outliers. In contrast, rank-based objectives remain stable because they depend on the sign of pairwise differences or the rank of targets, rather than absolute magnitudes. Second, standard deep neural networks trained via MSE frequently suffer from local miscalibration. By deferring scale recovery to a Stage~2 isotonic projection, \textsc{CAIRO} structurally forces the predictions to be strictly empirically auto-calibrated. This converts flexible but historically uncalibrated neural representations into reliable, locally unbiased statistical forecasters. Our framework suggests a re-evaluation of standard regression practices.

\paragraph{Limitations and future work.} While CAIRO is consistent at the population level, its finite-sample performance depends critically on the quality of the first-stage ranker. If the initial model fails to capture the true ordering of the conditional mean, no amount of monotonic calibration can recover the correct regression function. Additionally, the computational cost of pairwise ranking losses scales quadratically with batch size ($O(n^2)$), which may necessitate efficient sampling strategies or listwise approximations for very large datasets. Future work should explore more flexible calibration methods beyond global isotonic regression.   Finally, extending this two-stage decoupling to multivariate or structured output prediction remains an open and promising direction.

\bibliography{references}
\bibliographystyle{icml2025}
\appendix
\onecolumn

\section{Results and Proofs}
For clarity of exposition, our theoretical analysis assumes that the score distribution is continuous, i.e., $\mathbb{P}\!\big(g(X)=g(X')\big)=0$. However, this assumption is not strictly necessary. In the presence of ties, our results regarding the Gini covariance and Spearman's correlation remain valid provided that the standard cumulative distribution function is replaced by the \textit{mid-distribution function} \citep{parzen2004quantile}. This function is defined as $\widetilde{F}_g(t)
= \mathbb{P}\!\big(g(X)<t\big) + \frac{1}{2}\,\mathbb{P}\!\big(g(X)=t\big).$ Unlike the standard CDF, the mid-distribution ensures that the centered expectation property $\mathbb{E}\!\big[\widetilde{F}_g(g(X))\big]=1/2$ holds even when $g(X)$ is not continuous, thereby preserving the symmetry required for the covariance identities derived in Theorem~\ref{thm:equivalence}.

\subsection{\textbf{Proof of Composition Theorem 
(Theorem~\ref{composition_is_f_star_statement})}} \label{composition_is_f_star}

\begin{proof}
Let $g$ be the output from the first stage. Then $g = h(m^{\star})$ for a strictly increasing function $h$. We want to characterize the second-stage optimizer $f^{\star}$ such that 
\[
f^{\star} \in \arg\min_{t \in \mathcal{F}} \mathbb{E}\left[(Y - t(g(X)))^2\right].
\]
By adding and subtracting $m^\star$ in the objective and expanding the square, we obtain
\begin{align*}
    \mathbb{E}\!\left[(Y - t(g(X)))^2\right]
&= \mathbb{E}\!\left[(Y - m^\star(X) + m^\star(X) - t(g(X)))^2\right] \\
&= \mathbb{E}\!\left[(Y - m^\star(X))^2\right]
   + \mathbb{E}\!\left[(m^\star(X) - t(g(X)))^2\right] \\
&\quad + 2\,\mathbb{E}\!\left[(Y - m^\star(X))\,(m^\star(X) - t(g(X)))\right].
\end{align*}

The first term $\mathbb{E}\!\left[(Y - m^\star(X))^2\right]$ is independent of $t$, since it depends only on the true regression function $m^\star$.  For the last term, we apply the law of iterated expectations:
\[
\begin{aligned}
\mathbb{E}_{Y, X}\!\left[(Y - m^\star(X))\,(m^\star(X) - t(g(X)))\right]
&= \mathbb{E}_X\!\left[\mathbb{E}_{Y|X}\!\left[(Y - m^\star(X))\,(m^\star(X) - t(g(X))) \,\middle|\, X\right]\right] \\
&= \mathbb{E}_X\!\left[(m^\star(X) - t(g(X)))\,\mathbb{E}_{Y|X}\!\left[Y - m^\star(X) \,\middle|\, X\right]\right] \\
&= 0,
\end{aligned}
\]
since $\mathbb{E}_{Y|X}[Y - m^\star(X) \mid X] = 0$. Hence the objective simplifies to
\[
\arg\min_{t \in \mathcal{F}} \mathbb{E}\left[(Y - t(g(X)))^2\right]
= \arg\min_{t \in \mathcal{F}} \mathbb{E}\left[(m^{\star}(X) - t(g(X)))^2\right].
\]
This is minimized when $t(g(X)) = m^{\star}(X)$, i.e.,
\[
t(h(m^{\star}(X))) = m^{\star}(X) \quad\Longrightarrow\quad t = h^{-1}.
\]
The inverse exists since $h$ is strictly increasing, and $h^{-1} \in \mathcal{F}$ because it is non-decreasing as the inverse of a strictly increasing function. Therefore, the optimal second-stage function is $f^{\star} = h^{-1}$, and
\begin{equation*}
    f^{\star}(g(X)) = h^{-1}(h(m^{\star}(X))) = m^{\star}(X). \qedhere
\end{equation*}
\end{proof}

\subsection{Useful Results}
\label{oro_weighted}
\begin{lemma}
Let $(X,Y)$ and $(X',Y')$ be i.i.d.\ with finite second moments, then
\[
\mathbb{E}\!\left[(Y-Y')(X-X')\right] \;=\; 2\,\mathrm{Cov}(Y,X).
\]
\end{lemma}

\begin{proof}
\begin{align*}
\mathbb{E}\!\left[(Y-Y')(X-X')\right]
&= \mathbb{E}[YX]-\mathbb{E}[YX']-\mathbb{E}[Y'X]+\mathbb{E}[Y'X'] \\
&= \mathbb{E}[YX]-\mathbb{E}[Y]\mathbb{E}[X']-\mathbb{E}[Y']\mathbb{E}[X]+\mathbb{E}[Y'X'] \\
&= \mathbb{E}[YX]-\mathbb{E}[Y]\mathbb{E}[X]-\mathbb{E}[Y]\mathbb{E}[X]+\mathbb{E}[YX] \\
&= 2\big(\mathbb{E}[YX]-\mathbb{E}[Y]\mathbb{E}[X]\big)
= 2\,\mathrm{Cov}(Y,X).
\end{align*}
In the second line, independence across pairs gives
$\mathbb{E}[YX']=\mathbb{E}[Y]\mathbb{E}[X']$ and
$\mathbb{E}[Y'X]=\mathbb{E}[Y']\mathbb{E}[X]$.
In the third line, i.i.d.\ across pairs yields
$\mathbb{E}[Y'X']=\mathbb{E}[YX]$, $\mathbb{E}[X']=\mathbb{E}[X]$, and $\mathbb{E}[Y']=\mathbb{E}[Y]$.
\end{proof}

\begin{lemma}
Let $X$ and $X'$ be i.i.d.\ real random variables with continuous cdf $F$. With
$\operatorname{sgn}(t)=\mathbbm{1}\{t>0\}-\mathbbm{1}\{t<0\}$, we have
\[
\mathbb{E}\!\left[\operatorname{sgn}(X-X')\mid X\right]=2F(X)-1 \quad \text{a.s.}
\]
\end{lemma}

\begin{proof}
Starting with the decomposition $\operatorname{sgn}(X-X')=\mathbbm{1}\{X>X'\}-\mathbbm{1}\{X<X'\}$, we have
\begin{align*}
    \mathbb{E}[\operatorname{sgn}(X-X')|X] &= \mathbb{E}[\mathbbm{1}\{X>X'\}-\mathbbm{1}\{X<X'\}| X] \\
    &= \mathbb{P}(X>X'|X)-\mathbb{P}(X<X'| X) \\
    &= F(X) - (1-F(X)) \\
    &= 2F(X) - 1 . \qedhere
\end{align*}
\end{proof}

\begin{proposition}
\label{equivalent_pairwise_form}
Let $(X,Y)$ and $(X',Y')$ be i.i.d.\ with finite second moments and independent across pairs. Then
\[
\mathrm{Cov}\!\left(Y,\,F_{g(X)}(g(X))\right)
=\frac{1}{2}\,\mathbb{E}\!\left[(Y-Y')\,\mathbbm{1}\{\,g(X)>g(X')\,\}\right].
\]
\end{proposition}
\begin{proof}
Let $(X',Y')$ be an i.i.d.\ copy of $(X,Y)$.  Then 
\begin{align}
    2\cdot\text{Cov}\big(Y, F_{g(X)}(g(X))\big) \notag 
    &= 2\Big( \mathbb{E}\big[Y\, F_{g(X)}(g(X))\big] 
              - \mathbb{E}[Y]\,\mathbb{E}\big[F_{g(X)}(g(X))\big] \Big) \notag \\
    &= 2\Big( \mathbb{E}\big[Y\, F_{g(X)}(g(X))\big] - \tfrac{1}{2}\,\mathbb{E}[Y] \Big) \notag \\
    &= 2\,\mathbb{E}\big[Y\, F_{g(X)}(g(X))\big] - \mathbb{E}[Y] \notag \\
    &= 2\,\mathbb{E}\Big[ Y\, \mathbb{E}_{X'|Y,X}\big[\mathbbm{1}\{g(X') \le g(X)\}\mid Y, X\big]\Big] - \mathbb{E}[Y] \notag \\
    &= 2\,\mathbb{E}\big[ Y\,\mathbbm{1}\{g(X') \le g(X)\} \big] - \mathbb{E}[Y]. \label{eq:alt-start}
\end{align}
where we have use independence of $Y$ and $X'$ to condition on $Y$ in the inner expectation in line 4. We also applied the law of total expectation between the pair $(Y, X)$ and $X'$ to obtain the last line. Now, by symmetry of the i.i.d.\ pairs $(X,Y)$ and $(X',Y')$, 
\begin{align}
    2\,\mathbb{E}\big[ Y\,\mathbbm{1}\{g(X') \le g(X)\} \big]
    &= \mathbb{E}\big[ Y\,\mathbbm{1}\{g(X') \le g(X)\} \big]
     + \mathbb{E}\big[ Y'\,\mathbbm{1}\{g(X) \le g(X')\} \big]. \label{eq:alt-split}
\end{align}
We may replace $\le$ by $<$ in the indicators. Combining
\eqref{eq:alt-start} and \eqref{eq:alt-split}, we obtain
\begin{align}
    2\cdot\text{Cov}\big(Y, F_{g(X)}(g(X))\big)
    &= \mathbb{E}\big[ Y\,\mathbbm{1}\{g(X') < g(X)\} \big]
     + \mathbb{E}\big[ Y'\,\mathbbm{1}\{g(X) < g(X')\} \big]
     - \mathbb{E}[Y] \notag \\
    &= \mathbb{E}\big[ Y\,\mathbbm{1}\{g(X') < g(X)\} \big]
     + \mathbb{E}\big[ Y'\,\mathbbm{1}\{g(X) < g(X')\} \big] \notag \\
    &\quad - \mathbb{E}\big[ Y\big(\mathbbm{1}\{g(X') < g(X)\} + \mathbbm{1}\{g(X) < g(X')\}\big) \big] \notag \\
    &= \mathbb{E}\big[ Y'\,\mathbbm{1}\{g(X) < g(X')\} \big]
     - \mathbb{E}\big[ Y\,\mathbbm{1}\{g(X) < g(X')\} \big] \notag \\
    &= \mathbb{E}\big[ (Y' - Y)\,\mathbbm{1}\{g(X) < g(X')\} \big] \notag \\
    &= \mathbb{E}\big[ (Y - Y')\,\mathbbm{1}\{g(X) > g(X')\} \big].
\end{align}
This proves
\[
2\cdot\text{Cov}\big(Y, F_{g(X)}(g(X))\big)
= \mathbb{E}\big[ (Y - Y')\,\mathbbm{1}\{g(X) > g(X')\} \big],
\]
which is the desired identity.
\end{proof}

\subsection{Proof of Theorem~\ref{thm:equivalence}} \label{proof:equivalence_of_losses}
\begin{proof}
For any nonnegative symmetric weight $w(Y,Y')$ with $\mathbb{E}[w(Y,Y')]<\infty$, define
\begin{align}
\mathcal{L}(g) &\equiv\mathbb{E}\!\left[w(Y,Y')\,\mathbbm{1}\!\left\{(Y-Y')\,\operatorname{sign}\!\big(g(X)-g(X')\big)\le 0\right\}\right] \nonumber \\
    &= \mathbb{E}\!\left[w(Y,Y')\,\mathbbm{1}\!\left\{(Y-Y')\big(g(X)-g(X')\big)<0\right\}\right].
\label{base_eq}
\end{align}

\medskip
\noindent\textbf{1) Uniform weight and Kendall's $\tau$.}
Take $w(Y,Y')\equiv 1$, then $\mathcal{L}\xrightarrow{}\mathcal{L}_{\mathrm{uni}}$. Also apply the following transformation:

\begin{equation}
\label{eq: trick}
    \mathbbm{1}\!\left\{(Y-Y')\big(g(X)-g(X')\big)<0\right\}
=
\frac{1-\operatorname{sign}(Y-Y')\,\operatorname{sign}\!\big(g(X)-g(X')\big)}{2} .
\end{equation}
Plugging this into \eqref{base_eq} yields
\begin{align*}
\mathcal{L}_{\mathrm{uni}}(g)
=
\frac12
-\frac12\,\mathbb{E}\!\Big[\operatorname{sign}(Y-Y')\,\operatorname{sign}\!\big(g(X)-g(X')\big)\Big].
\end{align*}
The population Kendall correlation between $Y$ and $g(X)$ is
\begin{align*}
    \tau\big(Y,g(X)\big)
    &\equiv
    \mathbb{E}\!\Big[\operatorname{sign}(Y-Y')\,\operatorname{sign}\!\big(g(X)-g(X')\big)\Big].\\
    &= 1-2\cdot\mathcal{L}_{\mathrm{uni}}(g) .
\end{align*}

\medskip
\noindent\textbf{2) Absolute gap weight and the Gini covariance.}
Take $w(Y,Y')=|Y-Y'|$, then $\mathcal{L}\xrightarrow{}\mathcal{L}_{\mathrm{abs}}$. 
\begin{align}
\mathcal{L}_{\mathrm{abs}}(g)
&= \mathbb{E}\!\left[\,|Y-Y'|\,\mathbbm{1}\!\left\{(Y-Y')\big(g(X)-g(X')\big)<0\right\}\right] \nonumber \\
&= \mathbb{E}\!\left[\frac{|Y-Y'|}{2}
-\frac{|Y-Y'|\,\operatorname{sign}(Y-Y')\,\operatorname{sign}\!\big(g(X)-g(X')\big)}{2}\right] \nonumber \\
&= \frac12\,\mathbb{E}|Y-Y'|
-\frac12\,\mathbb{E}\!\Big[(Y-Y')\,\operatorname{sign}\!\big(g(X)-g(X')\big)\Big]
\label{eq:abs-master}
\end{align}
where we used the transformation of \eqref{eq: trick} in the second line and also that $|Y-Y'|\operatorname{sign}(Y-Y')=Y-Y'$ to go from second to third equality. Now write $\operatorname{sign}\!\big(g(X)-g(X')\big) = \mathbbm{1}\{g(X)>g(X')\}-\mathbbm{1}\{g(X)<g(X')\}$
Therefore, by plugging this in \eqref{eq:abs-master} and further manipulating by swapping the i.i.d.\ pairs $(X,Y)\leftrightarrow (X',Y')$ and using symmetry,
\begin{align}
\mathbb{E}\!\Big[(Y-Y')\,\operatorname{sign}\!\big(g(X)-g(X')\big)\Big]
&=
\mathbb{E}\!\Big[(Y-Y')\,\mathbbm{1}\{g(X)>g(X')\}\Big]
-\mathbb{E}\!\Big[(Y-Y')\,\mathbbm{1}\{g(X)<g(X')\}\Big] \nonumber \\
&= 2\,\mathbb{E}\!\Big[(Y-Y')\,\mathbbm{1}\{g(X)>g(X')\}\Big].
\label{eq:sign-to-indicator}
\end{align}
Invoking Proposition~\ref{equivalent_pairwise_form}, we combine it with \eqref{eq:sign-to-indicator} to get
\[
\mathbb{E}\!\Big[(Y-Y')\,\operatorname{sign}\!\big(g(X)-g(X')\big)\Big]
=
4\,\mathrm{Cov}\!\left(Y,\,F_g(g(X))\right).
\]
Plugging into \eqref{eq:abs-master} yields the claimed identity:
\[
\mathcal{L}_{\mathrm{abs}}(g)
=
\frac12\,\mathbb{E}|Y-Y'|
-2\,\mathrm{Cov}\!\left(Y,\,F_g(g(X))\right) .
\]
\medskip
\noindent\textbf{3) CDF-rank gap weight and Spearman's $\rho_S$.}
Take $w(Y,Y')=\big|F_Y(Y)-F_Y(Y')\big|$, then $\mathcal{L}\xrightarrow{}\mathcal{L}_{\mathrm{cdf}}$ where $F_Y$ is the CDF of $Y$.
Under the assumption, $F_Y$ is continuous and strictly increasing on the support of $Y$, so $\operatorname{ sign}(Y-Y')=\operatorname{sign}\!\big(F_Y(Y)-F_Y(Y')\big)$. Consequently, 
\begin{align}
\mathcal{L}_{\mathrm{cdf}}(g)
&= \mathbb{E}\!\left[\,|F_Y(Y)-F_Y(Y')|\,\mathbbm{1}\!\left\{(Y-Y')\big(g(X)-g(X')\big)<0\right\}\right] \nonumber \\
&= \mathbb{E}\!\left[\,|F_Y(Y)-F_Y(Y')|\,\mathbbm{1}\!\left\{\operatorname{sign}(Y-Y')\operatorname{sign}(g(X)-g(X')\big)<0\right\}\right] \nonumber \\
&=\mathbb{E}\!\left[\big|F_Y(Y)-F_Y(Y')\big|\,
\mathbbm{1}\!\left\{\operatorname{sign}\!\big(F_Y(Y)-F_Y(Y')\big)\neq \operatorname{sign}\!\big(g(X)-g(X')\big)\right\}\right] \nonumber \\
&=
\frac12\,\mathbb{E}\big|F_Y(Y)-F_Y(Y')\big|
-\frac12\,\mathbb{E}\!\Big[\big(F_Y(Y)-F_Y(Y')\big)\,\operatorname{sign}\!\big(g(X)-g(X')\big)\Big].
\label{eq:cdf-master}
\end{align}
\smallskip

\emph{First term.} $F_Y(Y)$ is $\mathrm{Unif}(0,1)$, so
\[
\mathbb{E}\big|F_Y(Y)-F_Y(Y')\big|
=
\int_0^1\!\!\int_0^1 |u-v|\,du\,dv
=
\frac13.
\]

\smallskip
\emph{Second term.} Apply the sign-to-indicator identity \eqref{eq:sign-to-indicator}
with $Y$ replaced by $F_Y(Y)$, and then apply Proposition~\ref{equivalent_pairwise_form} again, now with the random variable $F_Y(Y)$
in place of $Y$:
\begin{align}
    \mathbb{E}\!\Big[\big(F_Y(Y)-F_Y(Y')\big)\,\operatorname{sign}\!\big(g(X)-g(X')\big)\Big] &= 2\,\mathbb{E}\!\Big[\big(F_Y(Y)-F_Y(Y')\big)\,\mathbbm{1}\{g(X)>g(X')\}\Big] \nonumber \\
    &= 4\mathrm{Cov}\!\left(F_Y(Y),\,F_g(g(X))\right) . \nonumber
\end{align}
Substituting into \eqref{eq:cdf-master} yields
\[
\mathcal{L}_{\mathrm{cdf}}(g)
=
\frac12\cdot\frac13
-2\,\mathrm{Cov}\!\left(F_Y(Y),\,F_g(g(X))\right)
=
\frac16
-2\,\mathrm{Cov}\!\left(F_Y(Y),\,F_g(g(X))\right)
\equiv \frac16-\frac{\rho_S(Y,g(X))}{6} .
\]
\medskip
This completes the proof of all three claims.
\end{proof}

\subsection{Proof of Theorem \ref{thm:oro_all_three}: ORO property}
\label{proof_of_oro}

\begin{proof}
Fix $(x,x')$ and consider the conditional pairwise ranking loss
\[
\mathcal{L}_w(Y, g(X); x,x')
\;\equiv\;
\mathbb{E}\!\left[
w(Y,Y')\,
\mathbbm{1}\!\left\{(Y-Y')\big(g(x)-g(x')\big)<0\right\}
\;\bigm|\; X=x,\,X'=x'
\right],
\]
where the inequality is strict under (A1) (no ties in $Y$). The optimal choice is to pick the sign of $g(x)-g(x')$
that yields the smaller conditional risk; i.e.\ one should choose $g(x)>g(x')$ whenever
\begin{align}
&\mathbb{E}\!\left[
w(Y,Y')\,\mathbbm{1}\{Y<Y'\}\,\big|\, X=x,X'=x'
\right]
\le
\mathbb{E}\!\left[
w(Y,Y')\,\mathbbm{1}\{Y>Y'\}\,\big|\, X=x,X'=x'
\right] \nonumber \\
&\implies\quad
\mathbb{E}\!\left[
w(Y,Y')\big(\mathbbm{1}\{Y<Y'\}-\mathbbm{1}\{Y>Y'\}\big)\,\big|\, X=x,X'=x'
\right]
\le 0 \nonumber \\
&\implies\quad
\mathbb{E}\!\left[
w(Y,Y')\,\operatorname{sign}(Y-Y')\,\big|\, X=x,X'=x'
\right]
\ge 0 \, . \nonumber
\end{align}

If we define the quantity $S_w(x,x')
\equiv
\mathbb{E}[
w(Y,Y')\operatorname{sign}(Y-Y')
|\; X=x,X'=x'
],$
then the pairwise optimal decision rule is $\mathrm{sign}\!\big(g(x)-g(x')\big)
=
\mathrm{sign}\!\big(S_w(x,x')\big)$.  Consequently, to prove the ORO property it suffices to show that, for each of the three weights,
\begin{equation}
\label{eq:key_alignment}
\mathrm{sign}\!\big(S_w(x,x')\big)
=
\mathrm{sign}\!\big(m^\star(x)-m^\star(x')\big) .
\end{equation}
Indeed, if \eqref{eq:key_alignment} holds, then any population minimizer must order
pairs $(x,x')$ the same way as $m^\star$, which implies $g^\star=h( m^\star)$ for some strictly
increasing $h$, i.e.\ $g^\star\in\mathcal{G}'$. Conversely, any strictly increasing transform of $m^\star$ induces the same ordering and therefore achieves the same minimal conditional risk for all pairs $(x,x')$,
yielding global optimality.

\smallskip
\noindent\textbf{Case 1: $w(y,y')\equiv 1$}.
Here:
\begin{align*}
    S_{\mathrm{uni}}(x,x') 
    &= \mathbb{E}\!\left[\mathrm{sign}(Y-Y')\mid X=x,X'=x'\right] \\
    &= \mathbb{P}(Y>Y'\mid x,x')-\mathbb{P}(Y<Y'\mid x,x')\\
    &= \mathbb{P}(Y>Y'\mid x,x')-(1-\mathbb{P}(Y>Y'\mid x,x'))\\
    &= 2\,\mathbb{P}(Y>Y'\mid x,x')-1.
\end{align*}
where we have assumed the assumption of no ties in the third equation to sum both probabilities to 1. By assumption, $\mathbb{P}(Y>Y'\mid x,x')>\tfrac12$ if and only if
$m^\star(x)>m^\star(x')$, which gives \eqref{eq:key_alignment} for $w\equiv 1$.

\noindent\textbf{Case 2: $w(y,y') = |y-y'|$}.
Here:
\begin{align*}
S_{\mathrm{abs}}(x,x')
&= \mathbb{E}\!\left[\,|Y-Y'|\,\mathrm{sign}(Y-Y') \mid X=x,X'=x' \right] \\
&= \mathbb{E}\!\left[\,Y-Y' \mid X=x,X'=x' \right] \\
&= \mathbb{E}[Y\mid X=x] - \mathbb{E}[Y'\mid X'=x'] \\
&= m^\star(x) - m^\star(x').
\end{align*}
In the second line we used the identity $|t|\,\mathrm{sign}(t)=t$. Therefore
\[
\mathrm{sign}\!\big(S_{\mathrm{abs}}(x,x')\big)
=
\mathrm{sign}\!\big(m^\star(x)-m^\star(x')\big),
\]
so \eqref{eq:key_alignment} holds for the absolute-gap weight.

\smallskip
\noindent\textbf{Case 3: $w(y,y') = |F_Y(y)-F_Y(y')|$}.
Under Spearman's $\rho$ in case~\ref{Spearman}, $F_Y$ is continuous and strictly increasing, hence $\mathrm{sign}\!\big(F_Y(Y)-F_Y(Y')\big)=\mathrm{sign}(Y-Y')$, therefore:
\begin{align*}
S_{\mathrm{cdf}}(x,x')
&= \mathbb{E}\!\left[\,|F_Y(Y)-F_Y(Y')|\,\mathrm{sign}(Y-Y') \mid X=x,X'=x' \right] \\
&= \mathbb{E}\!\left[\,|F_Y(Y)-F_Y(Y')|\,\mathrm{sign}\!\big(F_Y(Y)-F_Y(Y')\big) \mid X=x,X'=x' \right] \\
&= \mathbb{E}\!\left[\,F_Y(Y)-F_Y(Y') \mid X=x,X'=x' \right] \\
&= \mathbb{E}\!\left[F_Y(Y)\mid X=x\right] - \mathbb{E}\!\left[F_Y(Y')\mid X'=x'\right] \\
&= \eta(x) - \eta(x'),
\end{align*}
where $\eta(x)\equiv \mathbb{E}[F_Y(Y)\mid X=x]$. In the third line we used the identity
$|t|\,\mathrm{sign}(t)=t$. It follows that
\[
\mathrm{sign}\!\big(S_{\mathrm{cdf}}(x,x')\big)
=
\mathrm{sign}\!\big(\eta(x)-\eta(x')\big),
\]
so the minimizers of $\mathcal{L}_{\mathrm{cdf}}$ are exactly the strictly increasing transforms
of $\eta$. Under the (sufficient) part of the assumption where $Y=m^\star(X)+\varepsilon$ with $\varepsilon$
independent of $X$ and with strictly increasing CDF, the family of conditional distributions
$\{Y\mid X=x\}$ is strictly stochastically ordered by $m^\star(x)$. Since $F_Y$ is strictly
increasing, the map $x\mapsto \eta(x)=\mathbb{E}[F_Y(Y)\mid X=x]$ is strictly increasing in
$m^\star(x)$, i.e.\ there exists a strictly increasing $h$ such that $\eta=h( m^\star)$.
Consequently,
\[
\mathrm{sign}\!\big(S_{\mathrm{cdf}}(x,x')\big)
=
\mathrm{sign}\!\big(m^\star(x)-m^\star(x')\big),
\]
and \eqref{eq:key_alignment} holds for the rank-gap weight as well. This completes the proof of all three claims.
\end{proof}

\subsection{Alternative Expression for the Empirical Weighted ranking loss}
\label{app:emp-oroloss-eq}

\begin{proposition}
\label{prop:emp-oroloss-eq}
Assume $\{(x_i,y_i)\}_{i=1}^n$ is such that there are no ties in $\{y_i\}$ and in $\{g(x_i)\}$, and $w_{ij}=w_{ji}\ge 0$ for all $i\ne j$. Then the empirical loss defined in \eqref{eq:emp-oroloss} satisfies
\[
    \ell(g)
    \;=\;
    \frac{1}{n(n-1)}\sum_{i\ne j}
    w_{ij}\,
    \mathbbm{1}\{y_i>y_j\}\,\mathbbm{1}\{g(x_i)<g(x_j)\}.
\]
\end{proposition}

\begin{proof}
The event $(y_i-y_j)\big(g(x_i)-g(x_j)\big)<0$ occurs if and only if exactly one of the two differences is reversed:
\[
\mathbbm{1}\!\left\{(y_i-y_j)\big(g(x_i)-g(x_j)\big)<0\right\}
=
\mathbbm{1}\{y_i>y_j,\,g(x_i)<g(x_j)\}
+\mathbbm{1}\{y_i<y_j,\,g(x_i)>g(x_j)\}.
\]
Plugging this into \eqref{eq:emp-oroloss} yields
\begin{align*}
\ell(g)
&= \frac{1}{n(n-1)}\sum_{1\le i<j\le n}
    w_{ij}\Big[
    \mathbbm{1}\{y_i>y_j,\,g(x_i)<g(x_j)\}
    +\mathbbm{1}\{y_i<y_j,\,g(x_i)>g(x_j)\}
    \Big].
\end{align*}

Now consider the quantity $S$ defined as:
\begin{align*}
    S
&\equiv \sum_{i\ne j}
    w_{ij}\,
    \mathbbm{1}\{y_i>y_j\}\,\mathbbm{1}\{g(x_i)<g(x_j)\}.\\
    &= \sum_{1\le i<j\le n}
\Big[
    w_{ij}\,\mathbbm{1}\{y_i>y_j,\,g(x_i)<g(x_j)\}
    +
    w_{ji}\,\mathbbm{1}\{y_j>y_i,\,g(x_j)<g(x_i)\} \Big] \\
    &= \sum_{1\le i<j\le n}
w_{ij}\Big[
    \mathbbm{1}\{y_i>y_j,\,g(x_i)<g(x_j)\}
    +\mathbbm{1}\{y_i<y_j,\,g(x_i)>g(x_j)\}
\Big].
\end{align*}
where we rewrote $S$ by grouping terms over unordered pairs $\{i,j\}$ with $i<j$ in the second line and used the fact that $y_j>y_i$ and $g(x_j)<g(x_i)$ are exactly the negations of $y_i>y_j$ and $g(x_i)<g(x_j)$ in the third. Comparing this with the expression for $\ell(g)$ above, we obtain
\[
\ell(g)
= \frac{1}{n(n-1)}\,S
= \frac{1}{n(n-1)}\sum_{i\ne j}
    w_{ij}\,
    \mathbbm{1}\{y_i>y_j\}\,\mathbbm{1}\{g(x_i)<g(x_j)\},
\]
which is exactly the claimed identity.
\end{proof}

\subsection{Rank based expression of the empirical Gini covariance loss} \label{gini_form}

\begin{proof}
We derive the rank-based expression for the empirical Gini-covariance objective used in Stage~1.  
Let $\widehat{\mathrm{Cov}}(\cdot,\cdot)$ denote the empirical covariance, $\widehat{\mathrm{Cov}}(a_i,b_i)\equiv \frac{1}{n}\sum_{i=1}^n (a_i-\bar a)(b_i-\bar b),
$ and define the empirical CDF of the scores by $\widehat F_{g(X)}(t)\equiv \frac{1}{n}\sum_{j=1}^n \mathbbm{1}\{g(x_j)\le t\}$. Note that: 
\begin{align}
\widehat F_{g(X)}(g(x_i))
&= \frac{1}{n}\sum_{j=1}^n \mathbbm{1}\{g(x_j)\le g(x_i)\}
 = \frac{\mathrm{rank}(g(x_i))}{n}.
\label{eq:emp_cdf_is_rank}
\end{align}
Plugging \eqref{eq:emp_cdf_is_rank} into the loss gives
\begin{align}
\ell_{\textup{abs}}(y,g(x))
&\equiv -2\,\widehat{\mathrm{Cov}}\!\Big(y_i,\frac{\mathrm{rank}(g(x_i))}{n}\Big) \notag\\
&= -\frac{2}{n}\,\widehat{\mathrm{Cov}}\!\big(y_i,\mathrm{rank}(g(x_i))\big) \notag\\
&= -\frac{2}{n}\cdot \frac{1}{n}\sum_{i=1}^n \big(y_i-\bar y\big)\Big(\mathrm{rank}(g(x_i))-\overline{\mathrm{rank}}\Big),
\label{eq:cov_expand}
\end{align}
where $\overline{\mathrm{rank}}=\frac{1}{n}\sum_{i=1}^n \mathrm{rank}(g(x_i))$.
Since $\{\mathrm{rank}(g(x_i))\}_{i=1}^n$ is a permutation of $\{1,\dots,n\}$, we have $\overline{\mathrm{rank}}=\frac{1}{n}\sum_{r=1}^n r = \frac{n+1}{2}.
$
Expanding \eqref{eq:cov_expand} and using $\sum_{i=1}^n (y_i-\bar y)=0$ yields
\begin{align}
\ell_{\textup{abs}}(y,g(x))
&= -\frac{2}{n^2}\sum_{i=1}^n \big(y_i-\bar y\big)\Big(\mathrm{rank}(g(x_i))-\tfrac{n+1}{2}\Big) \notag\\
&= -\frac{2}{n^2}\sum_{i=1}^n \big(y_i-\bar y\big)\,\mathrm{rank}(g(x_i)) 
    + \frac{2}{n^2}\cdot \frac{n+1}{2}\sum_{i=1}^n \big(y_i-\bar y\big) \notag\\
&= -\frac{2}{n^2}\sum_{i=1}^n \big(y_i-\bar y\big)\,\mathrm{rank}(g(x_i)) \notag\\
&= -\frac{2}{n^2}\left(\sum_{i=1}^n y_i\,\mathrm{rank}(g(x_i))
      -\bar y\sum_{i=1}^n \mathrm{rank}(g(x_i))\right) \notag\\
&= -\frac{2}{n^2}\left(\sum_{i=1}^n y_i\,\mathrm{rank}(g(x_i))
      -\bar y\cdot \frac{n(n+1)}{2}\right).
\label{eq:gini_rank_final}
\end{align}
The second term in \eqref{eq:gini_rank_final} is a constant with respect to $g$, and therefore, up to an additive constant and a positive scalar, minimizing $\ell_{\textup{abs}}$ is equivalent to minimizing
\[
-\sum_{i=1}^n y_i\,\mathrm{rank}(g(x_i)),
\]
which is the desired rank-based form.
\end{proof}

\end{document}